\documentclass[letterpaper, 10 pt, conference]{ieeeconf}  

\IEEEoverridecommandlockouts                              
\usepackage{amsmath,amssymb,latexsym,mathrsfs,amssymb} 
\usepackage{array}
\usepackage{hyperref}
\usepackage[dvipsnames]{xcolor}
\usepackage{xfrac}
\usepackage{algorithm}
\usepackage{algorithmic}
\let\labelindent\relax
\usepackage{enumitem}
\usepackage{comment}

\newcommand{\bz}[1]{\textcolor{blue}{#1}}

\title{\LARGE \bf
Resource Allocation in Electricity Markets with Budget Constrained Customers
}

\author{Lila Perkins and Baosen Zhang
\thanks{The authors are partially supported by NSF grant ECCS-1942326 and University of Washington Department of Applied Mathematics Boeing Fellowship.}
\thanks{L. Perkins is with the Department of Applied Mathematics, University of Washington, Seattle {\tt\small lfdp@uw.edu}}%
\thanks{B. Zhang is with the Department of Electrical and Computer Engineering, University of Washington, Seattle {\tt\small zhangbao@uw.edu}}%
}

\newtheorem{theorem}{Theorem}

\newtheorem{lemma}{Lemma}

\newcommand{\R}{\mathbb{R}}
\newcommand{\PP}{\mathcal{P}}

\newcommand{\st}{\textrm{s.t.}}

\begin{document}

\maketitle
\thispagestyle{empty}
\pagestyle{empty}

\begin{abstract}
In electricity markets, customers are increasingly constrained by their budgets. 
A budget constraint for a user is an upper bound on the price multiplied by the quantity.
However, since prices are determined by the market equilibrium, the budget constrained welfare maximization problem is difficult to define rigorously and to work with. 
In this letter, we show that a natural dual-ascent algorithm converges to a unique competitive equilibrium under budget constraints. Further, this budget-constrained equilibrium is exactly the solution of a convex welfare maximization problem in which each user's utility is replaced by a modified utility that splices the original utility with a logarithmic function where the budget binds. We also provide an explicit piecewise construction of this modified utility and demonstrate the results on examples with quadratic and square root utility functions.
\end{abstract}

\section{Introduction}
The  rising price of electricity has received significant attention from customers, policy makers, industry and academia~\cite{wiser2025factors,Penn2025}. In this paper, we focus on one aspect that has been brought to the forefront by the increase in prices, the fact that customer budgets are becoming binding constraints~\cite{Mitovich25,baker2023metrics}.  

As shown in several studies, electricity bills are growing as part of customer income~\cite{baker2023metrics,brown2020high}. In addition, customers are becoming much more cognizant of the amount of money they spend on electricity and will adjust their usage according to their budgets~\cite{mamkhezri2025out,borenstein2007customer}. This behavior is not restricted to residential users, since the cost of electricity for compute has entered the operational decision of large AI loads~\cite{liu2014pricing,chen2025electricity}. Therefore, it is of interest to model and understand electricity markets with budget constrained users. 

 



In this paper, we study an idealized model of electricity markets where each user derives some utility from consuming electricity, and the system incurs some cost based on the sum of the users' consumptions. The objective is to maximize social welfare, defined as the sum of the user utilities minus the system cost. This model has been widely used in the literature, both in the context of power systems and other resource allocation problems~\cite{kirschen2018fundamentals,walrand2008economic}. Prices, also called marginal costs, are the dual variables of power balancing constraints. A fundamental result is that they  support a competitive equilibrium if convexity conditions are satisfied~\cite{varian2014intermediate,berry2013economic}. 

Budget constraints take the form of an upper bound on the product of price and quantity of each user. However, it is not obvious how the budget constraint should be included in the resource allocation problem. Since prices in the system depend on the solutions (they are the dual variables), a budget constraint is somewhat self referential and is more well-defined once the problem has been solved. 

Several results have suggested that markets and games with budget constraints are difficult to study analytically~\cite{maddock1992estimating,caragiannis2018efficiency}, with the possibilities of nonconvexities and multiple equilibria. Consequently, some studies assume that users are very small and prices stay roughly constant~\cite{ogunjuyigbe2017user,lesic2018consumers}. However, given the rapid growth of large loads in the power system~\cite{lauby2025reliability}, understanding the relationship between consumption and price is important. In another direction, empirical studies of users with budget constraints have been conducted~\cite{parco2005two}. 

In this paper, we extend the resource allocation problem to rigorously include budget constraints. We show that under these constraints, with standard convexity assumptions, there exists a unique price that supports a competitive equilibrium. In addition, we provide an equivalent primal interpretation of the problem, by showing that budgets can be represented by modified utility functions and removing the need to explicitly consider prices in the constraints. We show that this modified primal problem is convex, thus allowing many existing results in resource allocation to be applied to budget constrained problems.



\section{Problem Formulation}
\subsection{Setup} \label{sec:setup}
Consider $n$ users with utility functions $u_i : \R_+ \to \R$, consumed quantities $x_i \geq 0$, and budgets $b_i \geq 0$. These users are served by a system operator with a cost function $C: \R_+ \to \R$. We assume that the cost depends only on the sum of the quantities. 

By introducing an auxiliary variable $y$ for total consumption, the welfare maximization problem is \begin{subequations} \label{eqn:primal}
\begin{align}
    (\PP) \quad \max_{x,y} & \sum_{i = 1}^n u_i(x_i) - C(y) \\
      \st \;\; & x_i \geq 0 \quad \forall i \\
      & \sum_{i = 1}^n x_i = y  \label{eqn:power_balance} \\
      & \lambda x_i \leq b_i \quad \forall i \quad \text{(budget)} \label{eqn:budget}
\end{align}
\end{subequations}
where the budget constraint $\lambda x_i \leq b_i$  involves some price $\lambda$ (with units of \$/quantity). We note that this problem is not actually well-defined, since it is not obvious what $\lambda$ should be. 

A seemingly natural choice for $\lambda$ seems to be interpreting it as the derivative of the cost function $C'(y)$, and replacing \eqref{eqn:budget} by $C'(y) x_i \leq b_i$. This choice turns out to be incorrect because it does not satisfy the economic interpretation of price: where it should support a competitive equilibrium between the system operator and the users~\cite{kirschen2018fundamentals}. It is easy to check that even when $u_i$'s are linear and $C$ is quadratic, if $\lambda$ in \eqref{eqn:budget} is replaced by $C'(y)$, the centralized optimal solution is not an equilibrium solution.\footnote{The problem is in general not convex, but the solution is not hard to compute in a two-player game with linear utilities and a quadratic cost.}  

Therefore, the correct interpretation of $\lambda$ should be the dual variable of the power balance constraint \eqref{eqn:power_balance}, as is in a problem without budget constraints. However, in this case, the budget constraint becomes somewhat self referential, since it depends on a dual variable that needs to be computed while solving the problem. 

The rest of the paper answers two questions: 1) Is there a procedure to solve \eqref{eqn:primal} and obtain a price $\lambda$ that is a competitive equilibrium? 2) Is there an equivalent primal problem to \eqref{eqn:primal} without the need to involve a dual variable in the constraint? For the first, we show that there is a unique equilibrium price, and it can be obtained through a modified version of the standard dual ascent algorithm. For the second, we show that there is an equivalent primal formulation, and the problem is convex if all $u_i$'s are concave and $C$ is convex.  



\textbf{Assumptions}
Throughout this paper, we assume that each $u_i$ is strictly concave and continuously differentiable on $\mathbb{R}_+$ and $C$ is strictly convex, coercive and continuously differentiable on $\mathbb{R}_+$. To avoid the trivial solution of $x_i^*=0$ for all $i$, we assume that there exists at least one user $i$ such that $u_i'(0) > C'(0)$.




\subsection{Price Equilibrium in the Unconstrained Case} \label{sec:unconstrained-price-equilibrium}
Before addressing the budget constrained case, we briefly review how prices arise and reach an equilibrium in the unconstrained case. This will in turn serve as a template to how we handle the budget constrained problem. 


Without budget constraints, the welfare maximization problem becomes 
\begin{subequations} \label{eqn:P_uncon}
\begin{align}
    \max_{x,y} \quad & \sum_{i=1}^n u_i(x_i) - C(y) \\
    \text{s.t.} \quad & x_i \geq 0 \quad \forall i \\
                      & \sum_{i=1}^n x_i = y \label{eqn:balance}
\end{align}
\end{subequations}
which is a convex optimization problem. 
Let $\lambda$ be the dual variable associated with the power balance constraint \eqref{eqn:balance}. The (partial) Lagrangian is \begin{align*}
    \mathcal{L}(x, y, \lambda) &= \sum_{i = 1}^n u_i(x_i) - C(y) + \lambda \left(y - \sum_{i = 1}^n x_i \right) \\
    &= \sum_{i = 1}^n -\big(\lambda x_i - u_i(x_i) \big) + \big( \lambda y - C(y)\big)
\end{align*} 
and its dual is
\begin{align}
    g(\lambda) &= \sum_{i = 1}^n -\min_{x_i \geq 0} \big\{ \lambda x_i - u_i(x_i) \big\} + \max_{y \geq 0} \big\{\lambda y - C(y) \big\}. \label{eqn:g(lam)}
\end{align}


The dual problem \eqref{eqn:g(lam)} is separable and can be solved using a standard distributed algorithm. Specifically, at a given price $\lambda^k$, each user independently solves $
    x_i(\lambda^k) = \arg \max_{x_i \geq 0} u_i(x_i) - \lambda^k x_i,$
the system operator solves $
    y(\lambda^k) = \arg \max_{y \geq 0} \lambda^k y - C(y),$
and the price is updated via dual ascent $
    \lambda^{k + 1} = \lambda^k + \alpha^k \left(\sum_{i = 1}^n x_i(\lambda^k) - y(\lambda^k) \right).$
This process converges to $\lambda^*$ under standard step-size conditions if $u_i$'s and $C$ satisfy the assumptions in Sec.~\ref{sec:setup}~\cite{bertsekas1997nonlinear,boyd2004convex}. 


\section{Budget-Constrained Competitive Equilibrium}
In this section, we present a framework for handling budget constraints. We show that a unique competitive equilibrium exists under budget constraints and derive an iterative algorithm to find it. Then in Sec.~\ref{sec:primal} we show that we can interpret the budget constrained problem as a modified primal problem. 


Our strategy to find the equilibrium price mirrors the unconstrained case. Note that given a price $\lambda$, the budget constrained user $i$'s optimization problem  becomes 
\begin{align*}
    x_i^*(\lambda) = \arg \max_{x_i \geq 0} u_i(x_i) - \lambda x_i \quad \st \quad \lambda x_i \leq b_i
\end{align*} 
Since $u_i$ is strictly concave, the unconstrained maximum is attained by $x_i^* = \max(0,(u_i')^{-1}(\lambda))$ from the first order conditions.  Then, the budget constraint $\lambda x_i \leq b_i$ caps the user's expenditure, resulting in the closed form 
\begin{align}
    x_i^*(\lambda) = \max \left\{0, \min\left\{ (u_i')^{-1}(\lambda), \frac{b_i}{\lambda} \right\} \right\} 
    \label{eqn:x(lam)_closed-form}
\end{align}
where, when the budget is binding, the user consumes $\frac{b_i}{\lambda}$ instead of the unconstrained maximum. 

The system operator's problem stays unchanged from the unconstrained case:
\begin{align}
    y^*(\lambda) = \arg \max_{y \geq 0} \lambda y - C(y) = (C')^{-1}(\lambda). \label{eqn:y(lam)}
\end{align}

\subsection{Algorithm}
We consider a natural extension of the iterative implementation for the unconstrained case to the budget-constrained case.
\begin{algorithm}[ht]
    \caption{Budget-Constrained Price Update}
    \label{alg:constrained-price-update}
    \begin{algorithmic}[1]
        \REQUIRE utility functions $u_i$, budgets $b_i$, cost function $C$, initial price $\lambda^0 > 0$, tolerance $\varepsilon > 0$
        \ENSURE 
        \FOR{$k = 0, 1, 2, \ldots$}
            \STATE user $i$ solves $
            x_i^k = \arg \max u_i(x_i) - \lambda^k x_i$, \\ $\st \quad x_i \geq 0, \quad \lambda^k x_i \leq b_i
            $
            \STATE system operator solves $y^k = \arg \max_{y \geq 0} \lambda^k y - C(y)$
            \STATE the price is update via $
            \lambda^{k + 1} = \lambda^k + \alpha^k \left( \sum_{i} x_i^k - y^k \right)
            $
            \IF{$|\lambda^{k + 1} -  \lambda^k| < \varepsilon$}
                \STATE \textbf{Return:} $\{x_i^k\}$, $y^k$, $\lambda^k$ 
            \ENDIF
        \ENDFOR
    \end{algorithmic}
\end{algorithm}
Note that the only modification for the unconstrained dual-ascent described in Sec. \ref{sec:unconstrained-price-equilibrium} is that each user's subproblem now includes their budget constraint. Further, since each user's action depends only on the given price as well as their own utility and budget, the system operator needs no knowledge of the users' utility functions. 

\subsection{Existence and Uniqueness of Equilibrium Price}
 
The main result of this section is that Algorithm \ref{alg:constrained-price-update} converges to a unique equilibrium from any given initial price. 

\begin{theorem}
    Under the assumptions on $u_i$ and $C$ in Sec. \ref{sec:setup}, there exists a unique price $\lambda^* \in (0, \infty)$ at which supply equals demand and the market clears. Furthermore, Algorithm~\ref{alg:constrained-price-update} converge to $\lambda^*$ from any initial condition $\lambda^0 > 0$ where $x_i^*(\lambda)$ and $y^*(\lambda)$ are defined in \eqref{eqn:x(lam)_closed-form} and \eqref{eqn:y(lam)} respectively. \label{thm:alg1-conv}
\end{theorem}

\begin{proof}
    The proof relies on the following lemma, which we state now and prove at the end.

    \begin{lemma}
    Under the assumptions on $u_i$ and $C$ in Sec. \ref{sec:setup}, the  function $f(\lambda) := \sum_i x_i^*(\lambda) - y^*(\lambda)$ \begin{enumerate}[label=(\roman*)]
        \item is continuous and strictly decreasing on $(0, \infty)$;
        \item $\lim_{\lambda \to \infty} f(\lambda) = - \infty$.
    \end{enumerate} 
    \label{lemma:price-dynamics-ftn}
\end{lemma}

First, we show the existence and uniqueness of an equilibrium price. 
By Lemma \ref{lemma:price-dynamics-ftn}, the price dynamics function $f(\lambda) := \sum_i x_i^*(\lambda) - y^*(\lambda)$ is continuous and strictly decreasing on $(0, \infty)$ with $f(\lambda) \to - \infty$ as $\lambda \to \infty$. Also, by assumption in Sec. \ref{sec:setup} that there exists at least one user $i$ such that $u_i'(0) > C'(0)$, then for sufficiently small $\lambda > 0$, $f(\lambda) > 0$ by strict monotonicity of $x_i^*(\lambda)$ and $y(\lambda)$. 
Thus, by the intermediate value theorem, there exists a unique $\lambda^* \in (0, \infty)$ such that $f(\lambda^*) = 0$.

Now, we show that Algorithm \ref{alg:constrained-price-update} converges to $\lambda^*$. We use the standard stochastic approximation approach~\cite{borkar2008stochastic} and treat the $\lambda$ update as a continuous time system $\Dot{\lambda} = f(\lambda)$.
Note that $\lambda^*$ is a fixed point of $\Dot{\lambda}$.

Define the Lyapunov function candidate \begin{align*}
    V(\lambda) = -\int_{\lambda^*}^\lambda f(\mu) \, d\mu
\end{align*}
First, note that since $f$ is continuous, $V$ is well-defined and continuously differentiable. Now, we verify that $V$ is a Lyapunov function for the fixed point $\lambda^*$. We have $V(\lambda^*) = 0$ by construction and $V(\lambda) > 0$ for all $\lambda \neq \lambda^*$ from the fact that $f$ is strictly decreasing. The time derivative of $V$ is $\Dot{V}= \frac{dV}{d\lambda}\Dot{\lambda} = \big(-f(\lambda)\big) f(\lambda) = -\big[ f(\lambda) \big]^2$,  and it is negative for all $\lambda \neq \lambda^*$ and $0$ for $\lambda^*$. 

It remains to prove Lemma \ref{lemma:price-dynamics-ftn}. We show (i) by inspecting the supply and demand sides separately. On the supply side, since $C$ is strictly convex, $C'$ is strictly increasing, and so $y^*(\lambda) = (C')^{-1}(\lambda)$ is continuous and strictly increasing by the inverse function theorem. Similarly, on the demand side, $(u_i')^{-1}$ is continuous and strictly decreasing. Since $\frac{b_i}{\lambda}$ is also continuous and strictly decreasing on $\lambda > 0$, then their pointwise minimum is continuous and strictly decreasing, and so $x_i^*(\lambda) = \max \left\{0, \min\left\{(u_i')^{-1}(\lambda), \frac{b_i}{\lambda}\right\} \right\}$ is continuous and non-increasing. Thus, $f = \sum_i x_i^*(\lambda) - y(\lambda)$ is continuous and strictly decreasing.

For (ii), note that each $x_i(\lambda) \leq \frac{b_i}{\lambda} \to 0$ as $\lambda \to \infty$ and so $\sum_i x_i(\lambda) \to 0$ as $\lambda \to \infty$ while  $y(\lambda) = (C')^{-1}(\lambda) \to \infty$ as $\lambda \to \infty$ since $C$ is coercive by assumption. Thus, $f(\lambda) \to - \infty$ as $\lambda \to \infty$.
\end{proof}

\section{Primal Interpretation via Modified Utility} \label{sec:primal}
 Theorem 1 guaranties the convergence to a unique equilibrium price $\lambda^*$, but it is through an iterative procedure. The primal problem in \eqref{eqn:primal} is not well defined since it depends on a price that needs to be found later. A natural question is whether we can interpret Algorithm 1 as solving some well-defined primal problem. 


We show that the equilibrium of Algorithm 1 is exactly the solution of an unconstrained welfare maximization problem with a modified utility function. Specifically, we show that there exist strictly concave, continuously differentiable utility functions $\hat{u}_i$ depending only on $u_i$ and $b_i$ such that 
\[
\arg \max_{x_i > 0} \hat{u}_i(x) - \lambda x = x^*(\lambda)
\]
for all $\lambda$, where $x_i^*(\lambda)$ is the budget-constrained demand defined in \eqref{eqn:x(lam)_closed-form}. 
This equivalence means that at every price $\lambda$, the allocation chosen by a budget-constrained user with utility $u_i$ is equal to that of an unconstrained user with the modified utility $\hat{u}_i$. 
At each iteration, Algorithm \ref{alg:constrained-price-update} computes user allocations, $x_i^*(\lambda^k)$, via their budget-constrained subproblems. By the above equivalence, its equilibrium is the unique solution to the modified unconstrained welfare maximization problem 
\begin{subequations} \label{eqn:modified-primal} 
\begin{align}
    (\hat{\PP}) \quad \max_{x,y} & \sum_{i = 1}^n \hat{u}_i(x_i) - C(y) \\
      \st \;\; & x_i \geq 0 \quad \forall i \\
      & \sum_{i = 1}^n x_i = y  \label{eqn:modified_power_balance}
\end{align}
\end{subequations}
This result is formalized in Theorem \ref{thm:primal-equiv} but first we construct the modified utility $\hat{u}$ explicitly. 

\subsection{Construction of Modified Utility}

Consider a single user with strictly concave, continuously differentiable utility $u$ and budget $b > 0$. Recall, the budget-constrained demand of the user at price $\lambda$ restricted to the positive domain is given by  (\ref{eqn:x(lam)_closed-form}): \[
x^*(\lambda) = \max \left\{ 0, \min\left\{ (u')^{-1}(\lambda), \frac{b}{\lambda} \right\}\right\}.
\]
The key observation is that the budget-binding demand $\frac{b}{\lambda}$ can be written as \begin{align*}
    \frac{b}{\lambda} = \arg \max_{x > 0} b \log x - \lambda x
\end{align*}
since $\frac{b}{\lambda}$ satisfies $\frac{d}{dx}(b\log x - \lambda x )\big|_{x = \frac{b}{\lambda}} = \frac{b}{\frac{b}{\lambda}} - \lambda = 0$ and $b \log x - \lambda x$ is strictly concave. Then, we define $\hat{u}$ so that its marginal utility equals $u'(x)$ where the budget does not bind and $\frac{d}{dx} b \log x = \frac{b}{x}$ where the budget does bind.

The crossover points between the two regimes (binding budget and slack budget) are the nonnegative solutions to \begin{align}
    u'(\Tilde{x}) = \frac{b}{\Tilde{x}}. \label{eqn:crossover-pts}
\end{align}
Let $0 = \Tilde{x}_0 < \Tilde{x}_1 <  \dots < \Tilde{x}_k < \Tilde{x}_{k + 1} = \infty$ denote these crossover points as well as the boundary conventions on $x$. (Note that the number of crossover points $k$ depends on the shape of $u$ and the value of $b$.) On each interval $(\Tilde{x}_{j - 1}, \Tilde{x}_j)$, either $u'(x) \leq \frac{b}{x}$ for all $x$ or $u'(x) > \frac{b}{x}$ for all $x$. Thus, the modified utility $\hat{u}$ is  defined piecewise as follows \begin{align}
    \hat{u}(x) = \begin{cases}
        u(x) + c_j & x \in (\Tilde{x}_{j - 1}, \Tilde{x}_j), \quad u'(x) \leq \frac{b}{x} \\
        b \log x + d_j & x \in (\Tilde{x}_{j - 1}, \Tilde{x}_j), \quad u'(x) > \frac{b}{x}
    \end{cases} \label{eqn:piecewise-formulization}
\end{align}
where the constants $c_j$ and $d_j$ are determined sequentially to ensure continuity of $\hat{u}$ at each crossover point $\Tilde{x}_j$ with $c_1 = d_1 = 0$ so that $\hat{u}(x) = \min\{u(x), b\log x\} $ on $(0, \Tilde{x}_1)$.

By construction, $\hat{u}$ is continuous everywhere and continuously differentiable on each interval $(\Tilde{x}_{j - 1}, \Tilde{x}_j)$. Furthermore, $\hat{u}$ is differentiable at each crossover point since from the slack budget side  $\hat{u}'(\Tilde{x}_j) = u'(\Tilde{x}_j) $ and from the binding budget side $ \frac{b}{\Tilde{x}_j} = \hat{u}'(\Tilde{x}_j)$ which are equal by \eqref{eqn:crossover-pts}. Therefore, $\hat{u}$ is continuously differentiable on $(0, \infty)$ and \begin{align}
    \hat{u}'(x) = \begin{cases}
    u'(x) & u'(x) \leq \frac{b}{x} \\
    \frac{b}{x} & u'(x) > \frac{b}{x}
\end{cases} = \min \left\{u'(x), \frac{b}{x} \right\}.
\label{eqn:u_hat'}
\end{align}
\subsection{Primal Equivalence Theorem}
We now state the main result of this section that the equilibrium of Algorithm \ref{alg:constrained-price-update} solves the modified welfare maximization problem $(\hat{\PP})$ defined in \eqref{eqn:modified-primal}. 

\begin{theorem}[Primal Equivalence]
Consider $n$ users with strictly concave, continuously differentiable utilities $u_i$ and budgets $b_i > 0$, served by a system operator with strictly convex, continuously differentiable cost $C$. Let $\hat{u}_i$ be the modified utility for user $i$ as defined in \eqref{eqn:piecewise-formulization}. Then, $(\hat{\PP})$ 
is a convex optimization problem whose unique optimal allocations $(\hat{x}^*, \hat{y}^*)$ and optimal dual $\hat{\lambda}^*$ equal the equilibrium allocations $(x^*, y^*)$ and price $\lambda^*$ from Algorithm \ref{alg:constrained-price-update}.
\label{thm:primal-equiv}
\end{theorem}

\begin{proof}
    The proof relies on an intermediate lemma which we state now and prove at the end. 

\begin{lemma}
    Let $u : (0, \infty) \to \R$ be strictly concave and continuously differentiable, and let $b > 0$. Define $\hat{u}$ as in (\ref{eqn:piecewise-formulization}). Then,
    \begin{enumerate}[label=(\roman*), leftmargin=2em]
        \item $\hat{u}$ is continuously differentiable on $(0, \infty)$,
        \item $\hat{u}$ is strictly concave on $(0, \infty)$, \item For all $\lambda > 0$, \[
        \arg\max_{x > 0} \hat{u}(x) - \lambda x = \max \left\{0 , \min \left\{ (u')^{-1}(\lambda), \frac{b}{\lambda} \right\}\right\}
        \]
    \end{enumerate} 
    \label{lemma:modified-u}
\end{lemma}

By Lemma \ref{lemma:modified-u}, each $\hat{u}_i$ is strictly concave and by assumption $C$ is strictly convex and coercive. Thus, the objective $\sum_i \hat{u}_i(x_i) - C(y)$ is strictly concave on the convex constraint set $\left\{(x, y) \in \R^n \times \R \; : x_i \geq 0, \; \sum_i x_i = y \right\}$ and so $(\hat{\PP})$ is a convex optimization that attains a unique optimum $(\hat{x}^*, \hat{y}^*)$. 

Further, since $(\hat{P})$ is convex and satisfies Slater's condition, strong duality holds. Hence, the optimal dual variable $\hat{\lambda}^*$ of the power balance constrained \eqref{eqn:modified_power_balance} satisfies the KKT conditions $\hat{u}_i'(\hat{x}_i^*) = \hat{\lambda}^*$ for every user $i$ and $C'(\hat{y}^*) = \hat{\lambda}^*$.
Thus, $\hat{x}_i^* = x_i^*(\hat{\lambda}^*)$ by Lemma \ref{lemma:modified-u}(iii) and $\hat{y}^* = y^*(\hat{\lambda}^*)$. Since $\lambda^*$ is the unique market clearing price resulting from Algorithm 1  by Theorem \ref{thm:alg1-conv}, we conclude that $\hat{\lambda}^* = \lambda^*$.

It remains to prove Lemma \ref{lemma:modified-u}. Property (i) follows directly from the construction of $\hat{u}$ and the matching of derivatives at the crossover points via \eqref{eqn:crossover-pts}. 

For property (ii), each piece of $\hat{u}$ is strictly concave (either $u + \textrm{const.}$ or $b \log x + \text{const.}$ both of which are concave). At each crossover point $\Tilde{x}_j$, continuous differentiability of $\hat{u}$ from property (i) ensures that $\hat{u}'$ has no upward jumps. Thus, $\hat{u}'$ is strictly decreasing on each piece and across each crossover point. By induction on $j$, we have that $\hat{u}$ is strictly concave.

For property (iii), by (i) and (ii), $\hat{u}(x) - \lambda x$ is strictly concave on $(0, \infty)$ and thus has a unique maximizer $x^*$ satisfying the first order condition $\hat{u}'(x^*) = \lambda$. By \eqref{eqn:u_hat'}, this means that $\min\left\{ u'(x^*), \frac{b}{x^*}\right\} = \lambda$. If $u'(x^*) \leq \frac{b}{x^*}$, then $u'(x^*) = \lambda$ so $x^* = (u')^{-1}(\lambda)$ and the budget is slack since $(u')^{-1}(\lambda) \leq \frac{b}{\lambda}$. Otherwise, $\frac{b}{x^*} = \lambda$ so $x^* = \frac{b}{\lambda}$ and the budget is binding since $(u')^{-1}(\lambda) > \frac{b}{\lambda}$. In both cases, $x^* = \min \left\{(u')^{-1}(\lambda), \frac{b}{\lambda} \right\}$.
\end{proof}


\subsection{Examples}
\subsubsection{Quadratic Utilities (two crossover points)} Consider $u(x) =  \beta x - \frac{\alpha}{2} x^2$ with budget $b$. The crossover equation $u'(\Tilde{x}) = \frac{b}{\Tilde{x}}$ becomes $\beta -\alpha \Tilde{x} = \frac{b}{\Tilde{x}}$, or equivalently, $\alpha \Tilde{x}^2 - \beta\Tilde{x} + b = 0$.  When $\beta^2 < 4 \alpha b$, there are no crossover points and $\hat{u} = u$. When $\beta^2 \geq 4 \alpha b$, there are two crossover points $\Tilde{x}_{lo} = \frac{\beta - \sqrt{\beta^2 - 4\alpha b}}{2 \alpha}$ and $\Tilde{x}_{hi} = \frac{\beta + \sqrt{\beta^2 - 4\alpha b}}{2 \alpha}$, and the modified utility is \[
\hat{u}(x) = \begin{cases}
    \beta x - \frac{\alpha}{2}x^2 & x \leq \Tilde{x}_{lo} \\
    b \log x + c & \Tilde{x}_{lo} \leq x \leq \Tilde{x}_{hi} \\
     \beta x - \frac{\alpha}{2}x^2 + d
\end{cases}
\]
where $c := u(\Tilde{x}_{lo})$ and $d := b \log \Tilde{x}_{hi} + c$ are chosen such that $\hat{u}$ is continuous. 

For concrete examples, Fig. \ref{fig:u_hat-quadratics} shows the original and modified utility curves for $u_1(x) = 3x - 0.1 x^2$ with budget $b_1 = 5$ and $u_2(x) = 5x - 0.5x^2$ with budget $b_2 = 4$.
\begin{figure}[ht]
    \centering
    \includegraphics[width=\linewidth]{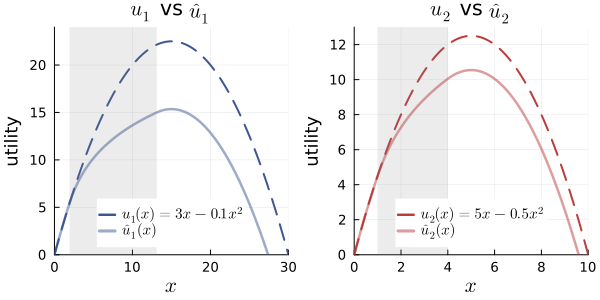}
    \caption{Original utilities $u_i$ (dashed) and modified utilities $\hat{u}_i$ (solid) for the two user quadratic examples. The modified utility replaces the original utility with $b \log x  + \textrm{const.}$ in regions where the budget binds (shaded gray), resulting in a slower-growing function which captures budget-constrained utility behavior.}
    \label{fig:u_hat-quadratics}
\end{figure}

\subsubsection{Square Root Utilities (one crossover point)}
Consider $u(x) = \alpha \sqrt{x}$ where $\alpha > 0$ and with budget $b$. The crossover equation becomes $\frac{\alpha}{2\sqrt{\Tilde{x}}} = \frac{b}{\Tilde{x}}$ which gives the single crossover point $\Tilde{x} = \frac{\alpha^2}{4b^2}$, and the modified utility is \[
\hat{u}(x) = \begin{cases}
    \alpha \sqrt{x} & x \leq \Tilde{x} \\
    b \log x + c & x > \Tilde{x}
\end{cases}
\]
where $c:= u(\Tilde{x})$.

For a concrete example, Fig. \ref{fig:u_hat-sqrts} shows the original and modified utility curves for $u_1(x) = 5\sqrt{x}$ with budget $b_1 = 3$. Fig. \ref{fig:u_hat-sqrts} also shows the original and modified utility curves for $u_2(x) = 10 \sqrt{x} - x$ with budget $b_2 = 4.5$ which can be thought of as the composition of a square root utility and quadratic utility. 

\begin{figure}[ht]
    \centering
    \includegraphics[width=\linewidth]{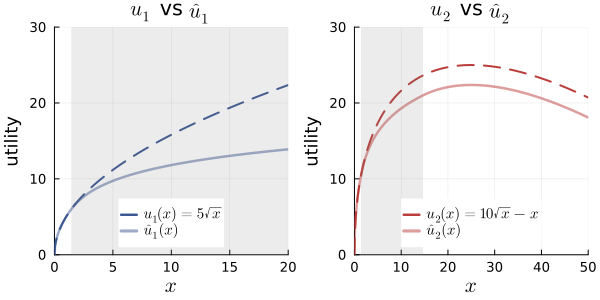}
    \caption{Original utilities $u_i$ (dashed) and modified utilities $\hat{u}_i$ (solid) for the two user square root examples. The modified utility replaces the original utility with $b \log x  + \textrm{const.}$ in regions where the budget binds (shaded gray), resulting in a slower-growing function which captures budget-constrained utility behavior.}
    \label{fig:u_hat-sqrts}
\end{figure}

\section{Simulations}
We validate our theoretical results numerically. In each example, we compare the solutions of the unconstrained problem $(\PP)$ and the modified problem $(\hat{\PP})$ from Theorem \ref{thm:primal-equiv} obtained via Algorithm \ref{alg:constrained-price-update}. Throughout, we use the cost function $C(y) = \frac{1}{2}y^2$.

\subsection{Two users with quadratic utilities} Consider $u_1(x) = 3x - 0.1 x^2$ with budget $b_1 = 5$ and $u_2(x) = 5 x - 0.5 x^2$ with budget $b_2 = 4$. Fig. \ref{fig:convergence} shows the convergence of Algorithm \ref{alg:constrained-price-update} to the equilibrium price $\lambda^*$. 
Fig. \ref{fig:two-quads-supply-and-demand} shows the supply and demand curves as functions of $\lambda$. Each user's budget-constrained demand $x_i^*(\lambda)$ (solid) is the minimum of the unconstrained demand $(u_i')^{-1}$ and the budget curve $\frac{b_i}{\lambda}$ (dashed). The equilibrium price $\lambda^*$ is the point where aggregate demand $\sum_i x_i^*(\lambda)$ equals supply $y^*(\lambda)$ verifying that $\lambda$ is a competitive equilibrium price. Table~\ref{tab:two-quadratic} shows how the optimal allocation and price change between the unconstrained and the budget constrained cases. 

\begin{figure}[ht]
    \centering
    \includegraphics[width=0.8\linewidth]{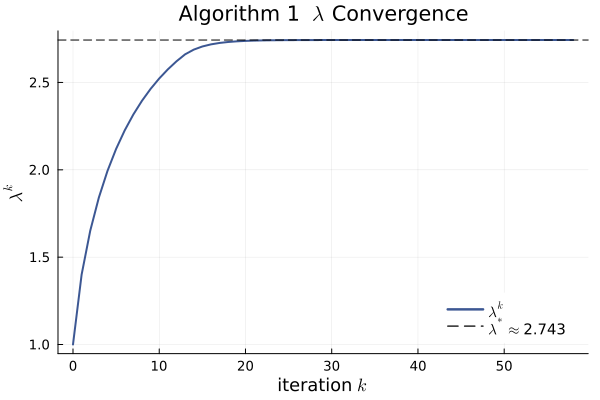}
    \caption{Convergence of the price iterate $\lambda^k$ in Algorithm \ref{alg:constrained-price-update} to the equilibrium price $\lambda^* = 2.743$ for the two-user quadratic utility example.}
    \label{fig:convergence}
\end{figure}

\begin{figure}[ht]
    \centering
    \includegraphics[width=0.8\linewidth]{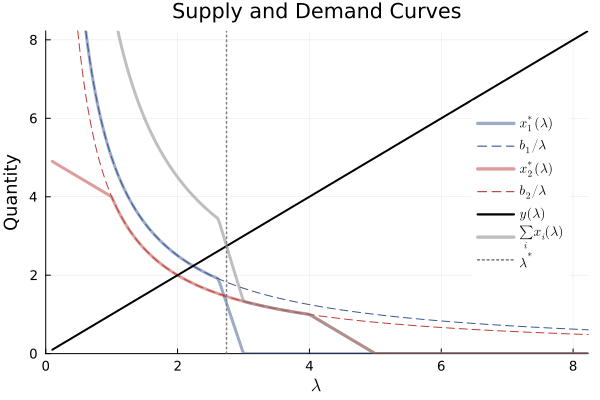}
    \caption{Supply and demand curves for the two-user quadratic utilities example. The competitive equilibrium price $\lambda^*$ is where aggregate demand (gray) intersects supply (black).}
    \label{fig:two-quads-supply-and-demand}
\end{figure}


\renewcommand{\arraystretch}{1.5}
\begin{table}[ht]
    \centering
    \caption{Equilibrium allocations and expenditures for the quadratic utility example}
    \label{tab:two-quadratic}
    \begin{tabular}{cc cc c cc  cc}
    \hline\hline
    && \multicolumn{2}{c}{Unconstrained} & & \multicolumn{2}{c}{Budget Constrained} & \\
    \cline{3-4} \cline{6-7} 
    && $x_i^*$ & $\lambda x_i^*$ & & $x_i^*$ & $\lambda x_i^*$ && $b_i$ \\
    \hline 
    User 1 && 0.714 & 2.041 & & 1.285 & 3.524 && 5.0 \\
    User 2 && 2.142 & 6.122 & & 1.458 & 4.000 && 4.0 \\
    \hline 
    $\lambda$ && \multicolumn{2}{c}{2.857} & & \multicolumn{2}{c}{2.743} && \\
    \hline \hline
\end{tabular}
\end{table}

\subsection{Five users with mixed utilities} 
To test the framework on more complex example, we consider five users with a mix of quadratic and square root utility functions with various budgets. Fig. \ref{fig:convergence-5} shows the convergence of Algorithm \ref{alg:constrained-price-update} to the equilibrium price $\lambda^*$. Table~\ref{tab:five-mixed} shows how the optimal allocation and price change between the unconstrained and the budget constrained cases.


\begin{figure}[ht]
    \centering
    \includegraphics[width=0.8\linewidth]{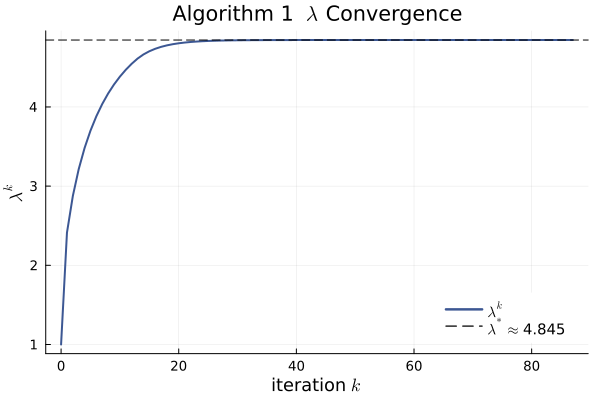}
    \caption{Convergence of the price iterate $\lambda^k$ in Algorithm \ref{alg:constrained-price-update} to the equilibrium price $\lambda^* = 4.845$ for the five-user mixed utility example.}
    \label{fig:convergence-5}
\end{figure}

\begin{table}[ht]
    \centering
    \caption{Equilibrium allocations and expenditures for 5 users}
    \label{tab:five-mixed}
    \begin{tabular}{cc cc c cc  cc}
    \hline\hline
    && \multicolumn{2}{c}{Unconstrained} & & \multicolumn{2}{c}{Budget Constrained} & \\
    \cline{3-4} \cline{6-7} 
    && $x_i^*$ & $\lambda x_i^*$ & & $x_i^*$ & $\lambda x_i^*$ && $b_i$ \\
    \hline 
    User 1 && 0.000 & 0.000 & & 0.310 & 1.500 && 4.0 \\
    User 2 && 1.530 & 8.370 & & 1.032 & 5.000 && 5.0 \\
    User 3 && 2.139 & 11.701 && 1.238 & 6.000 && 6.0 \\
    User 4 && 0.597 & 3.267 && 0.732 & 3.545 && 7.0 \\
    User 5 && 1.203 & 6.582 && 1.534 & 7.430 && 8.0 \\
    \hline 
    $\lambda$ && \multicolumn{2}{c}{5.470} & & \multicolumn{2}{c}{4.845} && \\
    \hline \hline
\end{tabular}
\end{table}

\section{Conclusions and Future Directions}
We studied resource allocation in electricity markets where customers face budget constraints on their expenditures. The main challenge is that budget constraints involve the equilibrium price, which is itself determined by the solution, making the problem self-referential and difficult to work with. 
We showed that a natural extension of the standard dual ascent algorithm (Algorithm \ref{alg:constrained-price-update}) converges to a unique competitive equilibrium from any initial price. 
The main contribution is the primal equivalence result (Theorem \ref{thm:primal-equiv}), specifically, the budget-constrained equilibrium is exactly the solution of the convex welfare maximization problem modified utility functions $\hat{u}_i$. These modified utilities are constructed piecewise as the original utility $u_i$ where the budget is slack and as $b_i \log x$ where the budget is binding. This construction applies to general concave utilities with any number of crossover points. 

A number of future directions remain. First, this paper considers a single-shot setting, where in practice the tariff schemes are often designed for long periods. 
Second, the modified welfare problem $(\hat{\PP})$ yields efficient allocations, it does not address fairness among users. Incorporating fairness criteria, such as proportional fairness, into the budget-constrained setting is a natural next step. 
Third, we consider a single-bus system with no network constraints. Incorporating transmission limits and nodal pricing would require extending the framework to more complex coupling constraints between users and vector-valued prices. 
Lastly, many budget-constrained customers are small and should be studied in aggregate. Modeling this cooperative behavior is within this budget-constrained framework an open problem.


\section{ACKNOWLEDGMENTS}

The authors would like to thank Ling Zhang and Yan Jiang for helpful discussions.


\bibliographystyle{IEEEtran}
\bibliography{refs.bib}

\end{document}